\documentclass[10pt, conference, compsocconf]{IEEEtran}
\usepackage{color}
\usepackage{amsthm}

\usepackage[normalem]{ulem}
\newtheorem{definition}{Definition}
\newtheorem{theorem}{Theorem}
\usepackage{amssymb}
\usepackage{comment}

\title{When Can a Distributed Ledger\\ Replace a Trusted Third Party?}

\author{
    \IEEEauthorblockN{Thomas Locher, Sebastian Obermeier, Yvonne Anne Pignolet}
    \IEEEauthorblockA{ABB Corporate Research\\Baden-D\"attwil, Switzerland\\\{thomas.locher, sebastian.obermeier, yvonne-anne.pignolet\}@ch.abb.com}
}

\begin{document}

\sloppy

\maketitle

\begin{abstract}
The functionality that distributed ledger technology provides, i.e., an immutable and fraud-resistant registry with validation and verification mechanisms, has traditionally been implemented with a trusted third party.
Due to the distributed nature of ledger technology, there is a strong recent trend
towards using ledgers to implement novel decentralized applications for a wide range of use cases,
e.g., in the financial sector and sharing economy.
While there can be several arguments for the use of a ledger, the key question is whether it can
fully replace any single trusted party in the system as otherwise a (potentially simpler) solution can be
built around the trusted party.
In this paper, we introduce an abstract view on ledger use cases and present two fundamental criteria
that must be met for any use case to be implemented using a ledger-based approach without having
to rely on any particular party in the system. Moreover, we evaluate several ledger use cases that
have recently received considerable attention according to these criteria, revealing that often participants
need to trust each other despite using a distributed ledger. Consequently, the potential of using a ledger
as a replacement for a trusted party is limited for these use cases.
\end{abstract}

\section{{Introduction}}
\label{sec:introduction}

There is tremendous interest in the blockchain, the distributed ledger that powers and secures the Bitcoin network. A plethora of applications and use cases for ledger technology have been suggested recently, ranging from registry services to smart contracts.
This interest is due to the desirable properties of distributed ledgers: Transactions are executed exactly once and an \emph{immutable} record of all transactions is maintained in a fault-tolerant and tamper-proof manner. Moreover, if the ledger is publicly available, as is the case for the Bitcoin ledger, anybody can verify the correctness of its records.
% any participant can validate the existence and provenance of the objects stored in the ledger---the \emph{bitcoins} in the Bitcoin network.

Traditionally, such features are offered through a trusted third party, which hosts multiple databases for the sake of availability and fault tolerance and vouches for the integrity of the stored data.
The main disadvantage of having a third party is that \emph{trust} is required in this party not to abuse its power and to faithfully provide its services. What is more, there is a risk that an attacker gains control over the third party, which enables the attacker to compromise the third party's services and invalidate its guarantees, e.g., the attacker could delete or modify records.

The main difference between using a third party and a ledger basically lies in the claim that the ledger removes the need to trust any particular party. In other words, trust is shifted from a specific party to a distributed system and its embedded protocols~\cite{mainelli2015sharing}. As a consequence, one needs Xto trust that the majority of the parties involved in maintaining the ledger follows the protocols, ensuring that the ledger operations are carried out as intended, and the remaining (malicious) entities cannot corrupt the system.
Since trust is a valuable and crucial commodity in any distributed system, it comes as no surprise that numerous use cases for ledger technology, other than virtual currencies, have been proposed and are being investigated.

Most proposed applications utilizing ledger technologies can be classified into three categories of growing complexity:$^3$
\begin{enumerate}
\item \emph{Registry service:} Storing digital records in an immutable and auditable distributed ledger.
\item \emph{Asset exchange:} Asset creation and ownership transfer.
\item \emph{Execute smart contracts:} Automate business processes through the execution of code.
\end{enumerate}
 The information stored in the ledger can represent physical or digital assets, identities, transactions, or contracts.  A protocol governs how entries are created, validated, recorded, and distributed.

For applications belonging to the first category, the ledger records important facts and events such as births, marriages, deaths, property deeds, intellectual property, election results, legal decisions, financial investments, insurance policies, or medical history. For such registry services, the main appeal is that records stored in the ledger are immutable and that they can potentially be used across organizational boundaries (with data protection and privacy mechanisms in place). Especially the latter is viewed as an important prerequisite for digitalization in the financial and medical industries as well as for governmental services.

Banks are particularly interested in the exchange of (digital) assets, facilitating cross-border payments and more transparent stocks, derivatives,  and options trading. Furthermore, transactions changing the ownership of or providing access to physical goods can be carried out on a ledger.

One step beyond transactions are so-called smart contracts, a (typically) distributed protocol that executes the terms of a contract autonomously with the aim of reducing the risk of error and manipulation.
It has been proposed to add support for smart contracts on top of a ledger: The contract is stored in the ledger in the form of executable code.
When a smart contract is executed, the ledger network members run the executable code according to the terms agreed upon in the contract. Since each execution starts with the same initial state, this automatic and distributed execution ensures consensus on the result among all members that execute the contract correctly.\footnote{If the code involves non-deterministic operations, shared coins~\cite{bellare1996distributed} can be used to guarantee that all parties maintain the same state.} Smart contracts offer the potential for new financial instruments, parameterized insurance contracts, and other services combining a shared database with the means for verifiable calculations or automated approval processes between two or more participants without trusted third parties.

While use cases are often described in terms of their potential compared to the state of the art, e.g., outlining potential cost savings, there is little discussion on how well the ledger actually fits the given use case. A noteworthy exception is the study by W\"ust et al.~\cite{wust17}, which introduces a methodology to determine if and what type of blockchain technology is appropriate for
a set of relevant use cases.

\textbf{Contributions.} In this paper, we focus on the specific question  whether a ledger can be used to transfer trust from key parties to a distributed system. To this end, we provide a model that is generic enough to capture a large set of use cases that have been proposed in practice and introduce criteria that any use case must necessarily meet in order to replace a trusted third party with a ledger-based solution. Furthermore, we discuss instantiations of our model for various use cases that have been implemented or proposed in recent years, showing that the criteria are not satisfied for most of them. This result implies that ledger-based technology may offer limited benefits for many currently proposed use cases.
It is worth noting that we do not address the question if and how  
incentives can lead the involved parties to act correctly and how a party can verify if a ledger implementation can be trusted. These challenges are out of the scope of this paper.

\section{{Background}}
\label{sec:background}
\textbf{Distributed Ledgers.}
The original blockchain as used in the virtual currency Bitcoin is a distributed
ledger for the serialization of transactions~\cite{nakamoto2008bitcoin}.
The ledger mechanisms provide a replicated linked list of immutable 
blocks, maintained by a large number of nodes, which 
tolerate malicious behavior of a small group of nodes and still reach consensus 
on the blocks of the ledger with a proof-of-work protocol. 

Using the ledger, virtual currency can be transferred from senders to receivers in a fully distributed 
manner, cutting out any middle man or trusted third party. This feature has gained enormous visibility 
and is envisioned to transform the financial sector and potentially bring disruptive innovation to many 
other sectors that traditionally rely on trusted third parties as well.\footnote{https://www.accenture.com/us-en/insight-highlights-insurance-blockchain-industrializing-trust, last accessed on April 9, 2018.}$^{,}$\footnote{http://www.economist.com/news/leaders/21677198-technology-behind-bitcoin-could-transform-how-economy-works-trust-machine, last accessed on April 9, 2018.}$^{,}$\footnote{https://dupress.deloitte.com/dup-us-en/focus/tech-trends/2017/blockchain-trust-economy.html, last accessed on April 9, 2018.} Among the ledger technology applications discussed in the past are supply chain provenance~\cite{kim2016towards,SupplyChain}, intellectual property rights 
management~\cite{swan2015blockchain}, proof of existence~\cite{crosby2016blockchain}, micro-payment systems~\cite{xu256blockchain}, to name but a few. Most of 
them rely on the immutability of the stored records and the fact that any party can 
verify if some properties or conditions hold. These applications are often built on top of systems that use other consensus protocols~\cite{swanson2015consensus}, such as proof-of-stake or classic Byzantine fault tolerance protocols~\cite{vukolic2015quest}.

More recently, the notion of smart contracts~\cite{szabo1997formalizing} has gained traction, with implementations in prominent projects such as
Ethereum~\cite{buterin2014next} and Hyperledger~\cite{cachin2016architecture}. Smart contracts are programs that are executed in the ledger environment and their correct execution is enforced by a consensus protocol. The rules of the contracts are encoded in a  programming language and determine under which conditions transactions or other operations are executed. For example, a contract can guarantee that virtual currency is only transferred once a certain event has happened, which has many applications, e.g., in insurance and escrow systems. 
If one of the parties breaches a contract or aborts a transaction, the 
ledger can ensure that the other involved parties obtain a compensation. 

Smart contracts are stored in the ledger and can be 
invoked by sending transactions to the contract's address. Once this transaction 
is added to the ledger, all blockchain nodes execute the contract code with 
the current state of the ledger and the transaction payload as input and 
agree on its output using the ledger's consensus 
protocol.\\

\noindent \textbf{Trusted Third Parties.}
Trust is based on the implicit or explicit assumptions that participants behave ''as they should'', i.e., that they follow a protocol as it was intended. The concrete assumptions on the abilities and behavior of (potentially malicious) participants varies from use case to use case and it is crucial to make these assumptions as clear as possible to choose the right technology and protocols. Hence, trust represents the flip-side of adversarial capabilities. For many interactions between participants, a trusted third party can expedite processes and ensure that malicious behavior is mitigated. In general, a trusted third party can execute arbitrary functions based on inputs provided by the participants.
Mainelli and Smith~\cite{mainelli2015sharing} describe three core functions a trusted third party must offer for ledger applications.
\begin{itemize}
\item \emph{Recording:} Holding the record of transactions in the event of dispute.
\item \emph{Transacting:} Preventing duplicate transactions, e.g., in order to prevent any party from selling the same thing twice (``double spending'').
\item \emph{Validating:} Confirming the existence of tradeable goods and membership of the trading community.
\end{itemize}
The first function can be implemented using standard database or distributed 
ledger technology. The second function requires additional 
mechanisms to linearize the transactions and verify the correctness of the linearization.
The implementation of the third function must be adapted for each specific use case and
can be hard to implement.
We will extend the definition of the third function in this paper and describe criteria to decide if ledger technology can provide validation for object existence, properties, conditions, and records.\\

\noindent \textbf{Reputation Systems.}
Distributed ledgers are not the only technology where trust is crucial. Trust is involved whenever two or more parties exchange virtual or physical goods.
\begin{comment}
In the simple case of Alice buying a computer from Bob's online shop, Alice wants to be sure that Bob is trustworthy and will send her a working computer and exchange it in case of problems. Bob, on the other hand, would like to be able to rely on Alice's credit worthiness. Being able to assess the trustworthiness of users is thus a prerequisite for any e-commerce system. Companies therefore have a high interest in establishing fair reputation systems and fight against false transactions and false identities created with the intent to maliciously manipulate reputation.
\end{comment}
A multitude of trust and reputation systems have been devised for various (web-based) applications, see~\cite{josang2007survey} for a survey. The authors describe different trust classes as well as categories for reputation and trust semantics that can be used in trust and reputation systems. Furthermore, they discuss centralized and distributed reputation system architectures and methods to compute reputation scores.

Theodorakopoulos and Baras~\cite{theodorakopoulos2004trust} investigate trust formed in interaction networks without a central trusted party, under the assumption that trust grows with each interaction where the participants have been following the protocol correctly.
\begin{comment}
 They model trust relationships as a graph, where nodes represent users, and edges represent direct trust relations, weighted by the amount of trust that the first user places on the second. Users have direct relations only toward the users they have interacted with in the path. They study trust inference schemes with the aim to establish an indirect relation between two users that have not previously interacted, assuming that trust is transitive, but in a way that takes edge weights into account.
\end{comment}
Recently, Carboni~\cite{carboni2015feedback}, and Schaub et al.~\cite{schaub2016trustless} presented approaches for reputation systems built on top of ledger technology. The former provides a reputation system using reputation vouchers that all participants involved in an exchange sign and that can be used to produce feedback on the trustworthiness of participants. The latter focuses on a privacy-preserving system applying blind signatures and a proof-of-stake ledger. 

It is important to note that all these reputation systems are based on the assumption that past behavior together with incentives to build further trustworthiness predicts how reliable a participant is, i.e., they cannot provide guarantees for future behavior. As a consequence, they may help to increase the trust in systems, but they cannot fully solve trust issues.

\section{{Model}}
\label{sec:model}

As mentioned before, there is significant interest in using distributed ledgers for a large range of use cases other than providing the underpinning for virtual currencies.
However, while the key characteristics closely match the requirements of recording global monetary transactions in an immutable ledger in an environment with open membership and no trust in any single entity in the system, applying distributed ledgers to other use cases is not straightforward.
In order to understand the limiting factors for a ledger-based solution with respect to trust, we must formally introduce a few concepts.

\subsection{Distributed Ledger}
\label{sec:distributed_ledger}
A distributed ledger $\Lambda$ represents a verifiable sequence of records, 
$S_k = (r_1, \ldots, r_k)$. 
The ledger is distributed in the sense that it is replicated across several machines operated by different entities in order to achieve fault-tolerance and protect against malicious actions. It
uses protocols, hashes, and digital signatures for validation, maintenance,  
authenticity, privacy, and access rights.\footnote{Protocols offering privacy and access control for distributed ledgers 
exist~\cite{kosba2016hawk,zyskind2015decentralizing} and may be required for 
some use cases. This does not affect the criteria we define.} In order to ensure the correct execution of protocols, they are also replicated and executed in a distributed manner where every correct
entity transitions to the same (correct) state. Since the correct entities share the same state, they agree on the same sequence of records $S_k=(r_1, \ldots, r_k)$. 
A set of participants (sometimes 
called users) may propose new records, while another (not necessarily disjoint) 
set of participants validates them, and adds 
them to the ledger following the aforementioned protocols.
We call participants in the second set \emph{maintainers}, which are called
validators, endorsers, orderers, or miners in different ledger approaches.

There are mechanisms in place that make it hard to change any information recorded in $\Lambda$ or 
at least make it easy to detect changes (audit trail). As a consequence, we 
assume the entries in $\Lambda$ to be \emph{immutable} with high probability, i.e., the probability that recorded items can be modified or removed from $\Lambda$ is negligible.
Thus, $\Lambda$ can provide the recording function of a trusted third party. 
As an example, the proof-of-work mechanism in Bitcoin ensures that successfully modifying blocks without disrupting the chain of hashes is highly improbable (unless an exorbitant amount of computational work is invested).
Hence the ledger protocols describe how records are created, 
validated, stored in the ledger, and distributed to all involved parties. 
More precisely, it defines a system of rules that allow the involved parties to 
exchange messages with agreed syntax and semantics.
In particular, we assume the ledger protocols guarantee the following 
properties for maintainers proposing records to be appended to the ledger. For simplicity, we let the index of the records $r_1, \ldots, r_k$ represent the \emph{logical time} at which they were added to the ledger.

\vspace{.2cm}
\begin{definition}[Ledger Consensus]\label{def:consensus}
A ledger consensus protocol has the following properties:\\
  \textit{P1 Agreement:} For a given $c\in \mathbb{N}$, if a correct maintainer holds a sequence of records $S_k=(r_1, \ldots, r_k)$, then for every correct maintainer that holds a sequence $S_j=(r'_1, \ldots, r'_j)$ we have that $r_i = r'_i$ for all $i \in \{1,\ldots, \min(j,k-c)\}$ for all $k\in \mathbb{N}$ with high probability. We say that the correct maintainers \emph{decided} on these records. \\
	\textit{P2 Validity:} A record $r_k$ decided by a correct maintainer must have been proposed by a participant.\\
	\textit{P3 Termination:} For all $k\in \mathbb{N}$, each correct maintainer eventually decides on a record $r_k$.
	\end{definition}
\vspace{.2cm}

The formal definition of consensus for distributed ledgers deviates from the classic consensus definition: Agreement is defined over a sequence and prefixes thereof. The motivation is that $\Lambda$ is considered an immutable list of records, and there may only be temporary inconsistencies with respect to the last few ($c \ge 1$) records. Decision is defined over the records that do not change anymore with high probability.

%Since it is an asynchronous system, there is no property guaranteeing 
%termination of the record creation process. In practice, one can often assume that new records are created within bounded time with high probability. Taking again Bitcoin as an example, a new block is generated roughly every ten minutes but in fact there is no guarantee that a new block is found in any time interval. 

As we do not impose any restrictions on how agreement is reached,
it is important to note that our model applies to both so-called 
\emph{permissioned} ledgers, where membership is governed by some potentially distributed authority (e.g., Hyperledger~\cite{cachin2016architecture}), and permissionless ledgers (such as the Bitcoin blockchain), where anybody can become a miner and propose blocks. 

As mentioned in the previous section, it is also possible to dynamically add smart contracts to $\Lambda$. As smart contracts encode instructions that are not part of the basic ledger functionality, the functionality of any particular smart contract may only be relevant for a subset of all participants.

\subsection{Use Case}

A \emph{use case} $U$ comprises the involved parties, the records that the 
parties create, and the objects involved in these records. Let $P$, $O$, and $R$ 
denote the set of (possible) parties, objects and records, respectively, in the 
considered use case~$U$. 

%We assume that each party $p \in P$ is equipped with a public/private key pair 
%and all public keys are known to all participants. 

Objects represent the basic elements considered in the use case. Records connect a logical time stamp (index), objects, and parties and are stored by a trusted third party or in a distributed ledger. Formally, a record $r_k \in R$ with index $k$ can be modeled as a tuple  $(k, P'\subseteq P, O' \subseteq O, s)$ where $s$ is a string describing the interaction of the parties and the objects. In general terms, the objective of a use case is achieved when a set of predicates over the records evaluates to \emph{true}. The exact nature of the predicates depends on the use case. 

%Some of the parties belong to the group of users proposing records to be 
% added to the ledger and some of the parties are maintainers of the ledger, or both.

We illustrate the meaning of the above definitions with a real-world example:
The most well-known use case of ledger-based technology is virtual currencies, 
such as Bitcoin. In Bitcoin, the bitcoins correspond to the objects, the 
parties represent the users and miners, and records are the transactions 
transferring bitcoins among the parties. An example of a predicate in this use case states that bitcoins are only transferred from user $A$ to user $B$ if user $A$ initiates this transaction. 

%The second example is Ethereum, a decentralized platform offering a virtual 
%currency (Ether) and running smart contracts. In this system, the parties 
%correspond to the users and miners of Ethereum. The smart contracts are additional protocols embedded in the ledger as
%defined above.
%The records are transactions and smart
%contracts, 
%and the objects are the objects referred to in the transactions and smart contracts.
%Note that the contracts are executed in the so-called Ethereum 
%Virtual Machine when triggered by a transaction with a smart contract as the destination 
%address. %In our notation such a transaction leads to a function record, which describes the result of the execution of the smart contract.

Recall that we focus on the implementation of a validation 
service through a distributed ledger as a replacement of a trusted third party. The specific properties and conditions to be verified depend on the use case, its objects, and the interactions of its parties.
Any object is characterized by its properties and functions. 
Note that the properties and functions of an object can be bound to 
certain parties, e.g., the current owner of an object can be saved as a reference to a party. As the sequence of records represent the state of the ''world`` of the use case, we can define predicates over single records and the whole sequence of records, which can be evaluated to \emph{true} or \emph{false}. 
The purpose of a validation service can be defined as follows.

\vspace{.2cm}
\begin{definition}[Validation Service] \label{def:validation}
A predicate $C$ over the set of records is a function mapping each record to a Boolean value, $C : R  \rightarrow \{true,false\}$.
Analogously, predicates can be defined over the sequence $S_k$ of the first $k$ records, $C_k: S_k \rightarrow \{true,false\}$. 
A validation service evaluates such predicates and provides the corresponding information to interested parties.  
\end{definition}
\vspace{.2cm}

The validation service can be used to verify if a predicate is satisfied by a single record or sequence of records. The predicates and their evaluation depend on the use case and the parties must agree on them. Note that the set of predicates that define the use case goals are typically a superset of the predicates that are validated by the validation service.

In practice, a validation services is used to \emph{confirm} the existence of 
an object and that all records involving this object conform with this 
object's properties and functions.

In the easiest case, this service is provided by a trusted third party that 
stores all records and performs all validation and verification actions.
As an example, a validation service for monetary transactions validates each 
transaction, checking that not more money is spent than the sender possesses.

When replacing the trusted third party with a distributed ledger $\Lambda$, 
its fixed protocols and the distributed nature of $\Lambda$ entail that 
\emph{no trust} is required that any particular entity adheres to the rules of 
the protocol. Rather, it is assumed that a majority of the parties execute the protocols correctly. Thus, a ledger can store records in a tamper-proof manner.
Moreover, the protocol can also facilitate validation by ensuring that all 
participants agree on the result of the corresponding procedures.
This functionality can often be implemented through smart contracts, which in turn can create 
more records to be added subsequently.

Note that a distributed ledger must meet the ledger consensus properties given in Definition~\ref{def:consensus} in order to implement a validation service. The agreement condition must hold as validation relies on the consistent information in the distributed ledger, i.e., all
correct maintainers must decide on the same records. Moreover, there must be some progress eventually (termination property P3) and only information proposed by participants must be entered (validity property P2). Furthermore, if the agreement property P1 only holds probabilistically for a ledger implementation, then the use case goal must be expressed in terms of this probability.

There are different ledger implementations for Definitions~\ref{def:consensus} and~\ref{def:validation}. In this paper we do not study how they are achieved and whether the incentive system of the ledger is appropriate for a given use case.  
Rather, we analyze conditions that must be met in addition to the ledger consensus properties in order to fully replace a trusted third party for any use case requiring a validation service.

%Having established a formal basis, we can define criteria for the applicability of blockchain-based technology at a more generic level, which is the subject of the subsequent section.

\section{{Criteria}}
\label{sec:criteria}

For the model defined above, we describe two general criteria that any use case must meet to provide a validation service through a distributed ledger without any trusted third party.
If any of the two (or both) criteria are not met, a ledger-based solution may only offer limited benefits and it is likely that other more efficient or more practical solutions exist.

\subsection{Object Creation Criterion}

\subsubsection{Definition}

Each party $p \in P$ is allowed to create records, and any existing object can 
be involved in the created records without restrictions. A fundamental 
question is where these objects come from. There are two cases that are 
''safe`` in the sense that no trust issues arise. In the simplest case, the 
set $O$ of objects is predetermined, i.e., each record $r \in R$ can only 
involve objects from a fixed set. The second case is the creation of new 
objects where the validation service confirms their existence. For example, 
new objects can be added to the sequence of records as a consequence of the predicate 
evaluation of a prefix of the sequence. %I.e., the protocol 
%underlying the system (and the ledger) is executed in a distributed manner to 
%agree on the creation of a new object. 
Hence, an agreed-upon predicate decides on the existence of an object in this case.

In order to implement the validation service using a distributed ledger in the second case, the protocols must ensure that the creation of objects is validated with an existence or creation predicate according to the properties described in Definition~\ref{def:consensus}, i.e., the protocol underlying the system (and
the ledger) is executed in a distributed manner to agree on
new objects. A creation process is called
consensus-based if it satisfies these properties. 

These cases are summarized in the following definition.

\begin{definition}[Object Creation Criterion]
Any use case $U$ meets the \emph{object creation criterion} if and only if
for all  $o\in O$ it holds that $o$ has been defined at $t=0$ or object creation is consensus-based.
\end{definition}

In order to understand the significance of this criterion, consider the case 
when the object creation criterion is violated. 
If the set $O$ is neither fixed nor determined by the underlying consensus 
protocol according to predicate evaluation, then parties can create records involving new objects individually. 
This power bears the potential to cause problems.

%\textbf{Denial-of-Service Attack.}
%If any party can create an arbitrary number of new objects, the system . Even if an upper bound on the number of created objects per party were somehow imposed, which requires the ledger to further keep track of the number of object creations for each party, this mechanism would be entirely ineffective if $P$ is not restricted, i.e., there can be an arbitrary number of parties: In this case, a \emph{Sybil attack}~\cite{douceur02} can be launched where a large number of (bogus) parties is created, circumventing the upper bound on the number of newly created objects. 
%Moreover, even if membership and object creation is limited, i.e., new parties and objects cannot freely be created, there can be other potentially more harmful attacks.
If the ledger keeps track of object ownership (which is the case for Bitcoin), malicious parties may try to \emph{anticipate} the creation of specific objects and enter records involving them prematurely into the ledger to falsely claim ownership. This attack is particularly troublesome when the identifier of future objects can be guessed based on the existing objects. The identifier may simply be a function of the objects characteristics potentially extended with a (predictable) serial number.

A possible scheme to protect against such premature object creation attacks is to ensure that there is enough randomness in the object identifiers to ensure that no future object (identifier) can be anticipated with overwhelming probability. However, even such a scheme can fall short of protecting the object creation process: Assume that multiple parties are involved in the creation of a new object, e.g., multiple employees in a company working on an instance of a product. Any party with knowledge about new objects and identifiers can prematurely add them to the ledger before the rightful owner has the chance to enter it. 
If $P$ is not restricted, the malicious insider can claim ownership for a newly created product under an arbitrary new identity, which may make it hard for the rightful owner to identify the culprit.

There is a simple solution to this problem: Restrict the object creation to specific parties. For example, an object referring to a specific product can only be entered by the party corresponding to the company that is known to produce this product. While this solution appears simple, it has a major downside. The parties must \emph{trust} specific parties to behave correctly. Returning to the example of the company selling certain products, a party must trust the party representing the company that, e.g., the objects it creates indeed correspond to (physical or digital) product instances. This is problematic as a key requirement for ledger-based systems is typically that no single entity must be trusted as mentioned earlier. Having to suddenly trust some party may invalidate the approach to use ledger-based technology in the first place because there are often much simpler solutions when trust is not an issue. Simply put, if we trust a party, it can simply maintain a database with all its records and 
report them upon request, without the need of a global and distributed ledger. The added benefit for the company is that it can further control access to its data, which is significantly harder to achieve when taking part in a global distributed system.

Instead of a specific set of trusted parties, it is possible to use a quorum-based approach
to create new objects in permissioned ledgers.
Quorum and voting based approaches are generally harder to implement in permissionless ledgers due to the fact that a malicious user can create an arbitrarily large number of (bogus)
parties voting according to this user's interest. In either case, the object creation criterion can be met if the participants agree on the predicate evaluation because the protocol controls the decision.

Thus, if we do not want to trust specific parties and object creation must be controlled, the system must meet the object creation criterion.
This observation is captured more formally in the following theorem.

\begin{theorem}
Any use case $U$ where a distributed ledger can replace a trusted third party meets the object creation criterion.
\end{theorem}
\begin{proof}
If the object creation criterion is not met, objects are not pre-defined and the creation is not consensus-based. This entails that some party can create some objects without the agreement of other parties, violating property P1 of Definition~\ref{def:consensus}. Consequently, the creation of objects is not validated through some pre-defined predicate in the system as in Definition~\ref{def:validation} and thus another mechanism is needed to validate them. If there is no trust in this validation mechanism, then there is no trust in the created objects, therefore trust in this mechanism is needed. Since this mechanism must be executed by some party, or a set of parties, trust in this party is required, i.e., an additional trusted (third) party, potentially comprising multiple entities, is needed.
\end{proof}

\subsection{Internal Predicate Criterion}

\subsubsection{Definition}
The second criterion is concerned with the fact that the stored records serve to indicate that some events (e.g., transactions) happened and certain conditions must hold. The purpose of the validation service in a ledger-based system is to use the stored records to 
evaluate some predicates over these records. The exact nature of the predicates depends on the use case. Using again the example of a virtual currency, a predicate may determine if a specific party holds a certain amount of currency,
which is a prerequisite for a monetary transaction.
If the truth value of a predicate $C$ can always be derived by examining only the current records stored in the ledger, %or the evaluation is consensus-based,
we say that $C$ is an \emph{internal} predicate. 
Any predicate that is not internal, regardless of whether the predicate can be verified at all, is called \emph{external}.

Given these definitions, we are now in the position to state the definition of the \emph{internal predicate criterion}.

\vspace{.2cm}
\begin{definition}[Internal Predicate Criterion] Any use case $U$ meets the \emph{internal predicate criterion} if and only if all predicates of $U$ are internal.
\end{definition}
\vspace{.2cm}

An example of an internal predicate in the context of virtual currencies is $C(r_i) = true$ for any record $r_i$ transferring funds from party $p$ to $p'$ if and only if the funds have indeed been transferred.
For this example, the predicate $C$ is trivially true as virtual currencies only exist as part of transactions and therefore the recording of a transaction is tantamount to transferring the currency.
Consequently, since only the current sequence of records must be examined to evaluate $C$, it is an internal predicate.

An example of a use case and associated predicate that does not meet the internal predicate criterion is the following. Assume that the ledger records cargo movement. The predicate for a particular record indicates 
whether a specific container has been shipped to a certain location. This predicate is external as the ledger cannot prove the position of a container in the physical world.

It is important to note that the internal predicate criterion is not necessarily violated whenever the predicate refers to objects and their state in the physical world.
In fact, it is possible to have predicates referring to physical objects that are internal, and there are also purely ``digital predicates'' that are external. As an example of the former, consider a system with a ledger that states ownership for any plot of land. For this clearly fictitious example, we assume that this ledger is globally recognized as the binding authority on this matter. Thus, if there is a transaction of selling a plot recorded in the ledger, the buyer will immediately be the rightful owner, i.e., the predicate is internal because the ledger alone is the authority on some aspect of the real world. 
As an example for the latter, consider the classic example of selling a website through a smart contract. The contract states that the website changes its owner once the new owner has transferred a predetermined amount to the current owner. However, the ownership only truly changes once the registrar updates the corresponding ownership record, i.e., it is an external predicate. The internal predicate criterion in this example could be met by getting rid of the registrar and, as in the preceding example, use the ledger as the binding authority. In this situation, the rightful owner (according to the ledger) can set up the website on any server and announce the website from this server publicly. Consequently, consistency between the ledger and external accounts must be guaranteed for external predicates, which requires trust in the consistency mechanisms.

Given that the use of a distributed ledger is typically bound to specific requirements, particularly, no trust in any single entity, the internal predicate criterion has fundamental implications: If a predicate is evaluated externally, its value must be provided by some party or parties, which may result in multiple undesirable situations. First, the party or parties may not be part of $P$, i.e., a completely \emph{different system} has a direct impact on the considered system. Problems in this situation can be that the other system changes unilaterally, which requires supervision and adaptation, or that the other system ceases to cooperate with the considered system or terminates its services altogether. Second, the considered system may depend on a single party, e.g., a specific organization or company. Thus, \emph{trust} is again required in this party to report the correct predicate value, which is, as stated before, often not desirable. 
As for the object creation criterion, violation of the internal predicate criterion is a strong indicator that a ledger-based solution may not be the right approach to implement a specific use case---and other options must be examined closely.

As for the object criterion, we can also state the significance of the internal predicate criterion more formally in the form of a simple theorem.

\begin{theorem} Any use case $U$ where a distributed ledger can replace a trusted third party meets the internal predicate criterion.
\end{theorem}
\begin{proof}
If the internal predicate criterion is not met, then there is at least one predicate $C$ that is not internal, i.e., it may not always possible to determine its truth value based on the records in the ledger alone. Consequently, there is some object property external to the distributed ledger that is required to evaluate $C$. In this case, a trusted party external to the distributed ledger system is needed to provide the missing information truthfully, otherwise the predicate cannot be evaluated.
\end{proof}

%In the subsequent section, we study and evaluate some proposed use cases for ledger technology that have received significant attention with respect to both criteria. 

\section{{Use Case Evaluation}}\label{sec:use_cases}

In this section, we investigate whether the criteria defined above are met for a selection of proposed use cases. Note that we do not intend to provide an exhaustive list of use cases but show how to apply the criteria to prominent use cases that have been proposed or implemented in recent years.

\subsection{Virtual Currency}

The prototypical use case of distributed ledgers is to enable a virtual currency such as Bitcoin. In this use case, all virtual currency is either ``pre-mined'', i.e., exists from the very beginning, or is created in each block that is mined, and transferred using transactions recorded in the ledger. Recall that in this use case the parties are the currency holders, the objects are the ``coins'', and the records are the transactions.

The object creation criterion is met in this case because either no new coins are ever created or new coins are created directly by the miners executing the ledger protocol, i.e., it is consensus-based.
As mentioned above, the quintessential predicate for virtual currencies asks whether some funds have been transferred in a transaction. We found that the internal predicate criteria is also met as it is internal, i.e., the ledger itself acts as the only valid source of information regarding ownership of funds.
Since there is no additional goal in this use case and both criteria are met, a distributed ledger can be used to replace a trusted third party as is evident from the success of Bitcoin.

\subsection{Notary Service for Ownership and Provenance}

A proposed use case for ledger-based technology is to provide a service to notarize ownership of physical objects or intellectual property.
There are multiple emerging companies offering provenance platforms for a variety of goods.
One example is tracking the ownership of diamonds: The idea is that diamonds can be \emph{uniquely} identified using more than 40 features, which can be used to create unique identifiers. Insurance companies add records of diamonds into their ledger to immutably record ownership. Such a notary service is interesting for insurance companies because it can help them recoup costs after paying out an owner when a diamond has been stolen or lost since 
the diamond belongs to the insurance company if it re-emerges. Moreover, this ledger should deter thieves from offering goods on online retail marketplaces as the ledger can identify the offered goods as stolen.

The participants of this use case are the insurance and mining/production companies and their clients, the objects are the diamonds or other luxury goods and the records represent proofs of ownership or provenance.

This use case does not meet the object creation criterion when parties can add new objects without a consensus-based protocol. This problem is addressed in that only insurance companies can create new objects. Since trust in the insurance companies is required anyway, entrusting the object creation process to the insurance companies is not a major concern. A feature to enable customers to create their own records is planned, which will violate the object creation criterion.

A specific problem for the diamond use case is that the argument of an economic disincentive for fraudsters to change the characteristics of a diamond is wrong: It is true that any physical change to the diamond results in a loss of value but a stolen diamond has no value at all if it cannot be sold. Therefore, there is an incentive to change the defining characteristics and register it as a new diamond. Adding a diamond to the ledger may require a proof of provenance, which may be 
hard to obtain for a fraudster. Thus, trust in mining companies to correctly certify provenance is required.

Several crucial predicates, e.g., asking whether an object has been lost or stolen, cannot be verified based on the distributed ledger. As discussed above, uniqueness cannot be determined either in the sense that the same object can be registered again with different characteristics. It follows that the internal predicate criterion is not met either. Given that both criteria are not met, it is essentially a system based on trust in certain parties, e.g., the insurance and mining companies, there is no clear incentive to use a ledger-based solution for this use case. Thus, a (distributed) database shared among the key parties
can implement exactly the same functionality.

\subsection{Inter-Bank Payments}

There is substantial interest in ledger-based technology in the financial industry as it is thought to simplify numerous complex processes inside banks and also among banks. We will focus on the impact that a ledger might have on inter-bank payments. Nowadays, money passes through several banks (each step incurring a certain fee) in a typical money transfer. The complexity of this process entails that a transaction often takes one or more days to process.
An obvious idea is to replace these complex mechanisms with a ledger-based approach, which would have the potential to significantly reduce transaction times, offer fault tolerance, traceability of transactions, and simpler reconciliation of accounts.

In this use case, the participants are banks and their clients. The objects are monetary units of a certain currency, and the records represent the transferring of money between accounts. 

Unlike the Bitcoin example, the object creation criterion is not met in this use case as the creation of money of non-virtual currencies is not consensus-based. Thus, trust in governments and banks is still required to handle the creation of money as it is done today. Naturally, this violation is not just acceptable but desirable for the banks since it ensures that they will remain an important party in the financial system. Regarding the crucial predicate asking whether funds have been transferred, the internal predicate criterion can be met if the content of the distributed ledger is considered to be (legally) binding. In this case, the accounts in all banks must be updated in accordance with the distributed ledger and any discrepancy must be identified and appropriate corrective actions taken to conform to the transaction history stored in the ledger.

Thus, ledger technology can be used to implement inter-bank payment systems but trust in the banks and governments is still required.
While it is technically feasible to use a ledger for this purpose, there are other challenges that must be overcome, in particular compliance with regulatory requirements and standards and the implementation of strong governance and data controls.

\subsection{Insurance Fraud} 

In the near future, insurance customers will be able to take pictures of broken goods and insurance claims forms with their phones and have them processed by a ledger registry~\cite{insurance}. This registry can be used by insurers to prevent customers from claiming payouts from multiple insurers. In the past, this could easily take weeks to investigate. A ledger registry could reduce the process of identifying duplicate claims to a matter of seconds.
More precisely, such a system could work as follows: customers make insurance claims through their insurer, where each claim record is time-stamped and signed by the system. The system inspects the ledger for duplicate claims. When combined with smart contracts, the claim can be settled immediately if the claim is deemed valid.

The participants in this use case are insurance companies and their customers. The objects are the insured goods and the records represent the claims. Note that in this use case the insurance companies have to trust each other and customers have to trust the insurers. However, the insurers do want to place trust in the customers. 

For the object creation criterion we need to distinguish between the insurance of specific objects and generic collections of objects (e.g., for renters/homeowners insurance or home contents insurance). In the second case, the only objects are contracts and claims for which the object creation criterion is met due to protocols. However, the first case does not meet the object creation criterion because the creation of objects for the insured goods suffers from the same issues as the notary service use case.

A much more important problem in this use case is the internal predicate criterion. The system must be able to evaluate the predicate asking whether a certain object is damaged or broken, which is not possible without a physical inspection. Pictures and videos do not provide sufficient proof and could be used by two different customers for claims of two instances of an object even though only one of them is broken. Moreover, a picture does not correspond to ownership and fraudulent claims could be made.

In summary, this ledger-based system cannot prevent insurance fraud in this case. In addition to the trust among insurance companies and from customers to the insurance companies, the existence of a trusted third party that verifies the existence and ownership of goods when establishing contracts and mechanisms to validate the claims is required.

\subsection{Microgrid Energy Trading}
In this use case, energy is traded using the ledger~\cite{energytrading}. Producers can offer a certain amount of energy that they produce, e.g., using solar panels, and a (smart) contract is settled with energy consumers. 

The participants are the energy providers, distributors and consumers, an object corresponds to a certain amount of energy and the records track the amount of energy produced, transmitted, and consumed as measured by smart meters.

The object creation criterion is not met because, in this scenario, trust in the smart meters reporting the energy production and consumption is required. A hacked firmware of one or multiple smart meters can report higher or lower values, which means the system is based on trusting the reported values.
For the same reason the predicate asking whether a certain amount of energy has been produced cannot be verified using the distributed ledger itself: The layout of the power grid can be maliciously modified such that two smart meters are installed sequentially instead of in parallel, which would double the reported production. Therefore, an inspection of the physical conditions is required, which means the system is based on trusting a party and its inspection report. As a result, the internal predicate criterion is not met.

While energy trading appears to be an interesting use case for a ledger for the purpose of automating payments based on reported energy values, it does not replace third parties and trust in the overall system.

\subsection{Supply Chain Management}
An often discussed use case is to employ a ledger to verify supply chains~\cite{SupplyChain,kim2016towards}. In this scenario, it is assumed that all suppliers use the ledger for material tracking and verification of certain attributes during each step of production and transportation.

Hence, the participants are the involved parties on the supply chain, the objects are the materials and goods produced, and the records represent the exchange of assets or modifications of assets and their properties.

It depends on the process implemented for adding goods into the ledger whether the object creation criterion is met: If it is consensus-based leading to an agreement among the involved parties, it can be satisfied, otherwise it requires trust. 

The internal predicate criterion is more problematic. First, the link to the material must be verified and it must be ensured that the object referred to in the ledger exists, i.e., the predicate ``does the object exist?'' must be evaluated. Furthermore, a supplier needs to state certain properties of an object, each property corresponding to a predicate, and enter this information into the ledger. However, these properties must be (physically) verified, and the entity that verifies this information must be trusted, which means the predicate is external. Even when electronic devices are used to verify properties, these devices are susceptible to hacks and modifications and therefore must be trusted.
In summary, supply chain verification requires trust in specific parties even when using a distributed ledger.

\subsection{Location Tracking}
 
Another proposed use case is to track containers on a ship using a distributed ledger.\footnote{https://www.technologyreview.com/s/603791/the-worlds-largest-shipping-company-trials-blockchain-to-track-cargo/, accessed on April 9, 2018.} In this scenario, the ledger is used to track all relevant information regarding a ship's cargo in an integrity-preserving manner. It is claimed that the ledger will allow the carrier to improve its processes by eliminating paperwork. Furthermore, smart contracts could be setup to automatically trigger payments after the cargo has arrived at its destination. 

The participants in this use case are the companies responsible for shipping and the owners of the cargo. The objects are the containers and the records represent the location of the objects at a certain point in time.

As for supply chain management, it depends on the properties of the object creation 
process whether the object creation criterion is satisfied. The first issue with respect to the internal predicate criterion is that it must be possible to verify that a container object inserted into the ledger corresponds to a real physical container. There must be trust in the entity that verifies this information as the ledger itself cannot perform this validation. Moreover, the internal predicate criterion is also not addressed later in this scenario: Neither the ledger nor a standard electronic document management approach solves the problem that an authenticated link between the physical container content with the electronic record must be established, which implies that the predicate asking whether a cargo load contains a specific object is external. Thus, there must be a third party verifying that the initial cargo load complies with the record. If the container might be opened during transport, e.g., for inspection, a second problem arises, which is to prove that nothing has been added, removed, or altered during this inspection. Finally, the entity reporting that the container has arrived at the destination must be trusted as GPS signals can be forged.
%Given the description of the scenario, electronic document management systems seem to address the problem well~\cite{adam2007implementing}. 
%The concept of decentralized trust management was already described in 1996~\cite{Decentralizedtrustmanagement}.
Thus, we can conclude that %this use case is fundamentally a data management problem and that 
trust towards some entities is crucial and cannot be removed by a ledger-based approach.

\section{{Conclusions}}

While the distributed and consensus-based nature of ledgers ensures immutability of records and fault tolerance without requiring trust among the participants, how these properties extend to trust in ledger use cases in a generic and formal manner has not been studied.

We have presented two fundamental criteria to answer the question whether a use case can be implemented without trusting any specific party using a ledger-based approach. In particular the criteria help to identify under which conditions trust can be technically guaranteed by ledger-based approaches and where additional, potentially non-technical mechanisms, such as legal frameworks are necessary. Our criteria reveal that using a ledger does not solve the trust issue for many use cases implemented or proposed in recent years.  
Moreover, the criteria bring to light that the ledger only becomes truly powerful when
it is recognized as the supreme authority in that its consensus protocol controls the object creation process and the predicate verification process is internal.

Regarding the internal predicate criterion, it is sometimes technically possible to turn an external predicate into an internal predicate. E.g., web domain registrars could be fully replaced by a ledger-based system in the future. For some use cases the organizational hurdles might be too high to overcome, e.g., due to conflicting legal frameworks. It is important to note that in some use cases it is impossible to replace trust by a ledger: it will always be necessary to involve another party for the verification of certain predicates. In general, predicates depending on measurements of a sensor will remain external because the sensor is a physical object that cannot be represented entirely in a digital system.

As a consequence, while ledger technology bears potential for new applications, many use cases will still require a trusted third party.

\end{document}